\makeatletter
\def\input@path{{styles/}{../styles/}}
\makeatother

\ifx\SODA\undefined%
\documentclass[11pt]{article}%
\providecommand{\SODAVer}[1]{}%
\providecommand{\NotSODAVer}[1]{#1}%
\else%
\documentclass[11pt]{article}%
\providecommand{\SODAVer}[1]{#1}%
\providecommand{\NotSODAVer}[1]{}%
\fi

\def\UseBibLatex{1}

\ifx\sarielComp\undefined%
\newcommand{\SarielComp}[1]{}
\newcommand{\NotSarielComp}[1]{#1}%
\else
\newcommand{\SarielComp}[1]{#1}%
\newcommand{\NotSarielComp}[1]{}%
\fi
\newcommand{\IfPrinterVer}[2]{#2}%

\NotSODAVer{\usepackage[cm]{fullpage}}%
\SODAVer{\usepackage[in]{fullpage}}%
\usepackage{amsmath}%
\usepackage{amssymb}%
\usepackage{xcolor}%
\usepackage{scalerel}%
\usepackage{euscript}%

\SarielComp{\usepackage{sariel_colors}}%

\usepackage[amsmath,thmmarks]{ntheorem}%
\theoremseparator{.}%

\usepackage{titlesec}%
\titlelabel{\thetitle. }%

\titleformat{\paragraph}[runin]
{\normalfont\bfseries}
{\theparagraph}
{1em}
{\addperiod}

\newcommand{\addperiod}[1]{#1.}

\usepackage{xcolor}%
\usepackage{mleftright}%
\usepackage{xspace}%
\usepackage{hyperref}%
\usepackage{mathcalb}

\IfPrinterVer{%
   \usepackage{hyperref}%
}{%
   \usepackage{hyperref}%
   \hypersetup{%
      breaklinks,%
      ocgcolorlinks, colorlinks=true,%
      urlcolor=[rgb]{0.25,0.0,0.0},%
      linkcolor=[rgb]{0.5,0.0,0.0},%
      citecolor=[rgb]{0,0.2,0.445},%
      filecolor=[rgb]{0,0,0.4},
      anchorcolor=[rgb]={0.0,0.1,0.2}%
   }
}

\theoremseparator{.}%

\theoremstyle{plain}%
\newtheorem{theorem}{Theorem}[section]

\newtheorem{lemma}[theorem]{Lemma}

\newtheorem{observation}[theorem]{Observation}

\theoremstyle{plain}%
\theoremheaderfont{\sf} \theorembodyfont{\upshape}%
\newtheorem*{remark:unnumbered}[theorem]{Remark}%
\newtheorem{remark}[theorem]{Remark}%

\newtheorem{defn}[theorem]{Definition}

\newtheorem{problem}[theorem]{Problem}

\newcommand{\myqedsymbol}{\rule{2mm}{2mm}}

\theoremheaderfont{\em}%
\theorembodyfont{\upshape}%
\theoremstyle{nonumberplain}%
\theoremseparator{}%
\theoremsymbol{\myqedsymbol}%
\newtheorem{proof}{Proof:}%

\definecolor{blue25emph}{rgb}{0, 0, 11}
\providecommand{\emphic}[2]{%
   \textcolor{blue25emph}{%
      \textbf{\emph{#1}}}%
   \index{#2}}

\definecolor{almostblack}{rgb}{0, 0, 0.3}
\providecommand{\emphw}[1]{{\textcolor{almostblack}{\emph{#1}}}}%

\providecommand{\emphi}[1]{\emphic{#1}{#1}}

\definecolor{almostblack}{rgb}{0, 0, 0.3}
\providecommand{\emphw}[1]{{\textcolor{almostblack}{\emph{#1}}}}%

\newcommand{\atgen}{\symbol{'100}}
\newcommand{\SarielThanks}[1]{\thanks{Department of Computer Science;
      University of Illinois; 201 N. Goodwin Avenue; Urbana, IL,
      61801, USA; {\tt sariel\atgen{}illinois.edu}; {\tt
         \url{http://sarielhp.org/}.} #1}}
\newcommand{\DavidThanks}[1]{%
   \thanks{%
      Department of Computer Science; %
      University of Illinois; %
      201 N. Goodwin Avenue; %
      Urbana, IL, 61801, USA; %
      {\tt dwzheng2\atgen{}illinois.edu}; %
      {\tt \url{https://davidzheng.web.illinois.edu/}.} #1}}

\newcommand{\HLink}[2]{\hyperref[#2]{#1~\ref*{#2}}}
\newcommand{\HLinkSuffix}[3]{\hyperref[#2]{#1\ref*{#2}{#3}}}

\newcommand{\figlab}[1]{\label{fig:#1}}
\newcommand{\figref}[1]{\HLink{Figure}{fig:#1}}

\newcommand{\thmlab}[1]{{\label{theo:#1}}}
\newcommand{\thmref}[1]{\HLink{Theorem}{theo:#1}}

\newcommand{\problab}[1]{\label{prob:#1}}%
\newcommand{\probref}[1]{\HLink{Problem}{prob:#1}}%
\newcommand{\probrefY}[2]{\hyperref[prob:#1]{#2}}

\newcommand{\seclab}[1]{\label{sec:#1}}
\newcommand{\secref}[1]{\HLink{Section}{sec:#1}}

\newcommand{\itemlab}[1]{\label{item:#1}}
\newcommand{\itemref}[1]{\HLinkSuffix{}{item:#1}{}}

\newcommand{\lemlab}[1]{\label{lemma:#1}}
\newcommand{\lemref}[1]{\HLink{Lemma}{lemma:#1}}%

\providecommand{\eqlab}[1]{}%
\renewcommand{\eqlab}[1]{\label{equation:#1}}
\newcommand{\Eqref}[1]{\HLinkSuffix{Eq.~(}{equation:#1}{)}}

\newcommand{\remove}[1]{}%
\newcommand{\Set}[2]{\left\{ #1 \;\middle\vert\; #2 \right\}}
\newcommand{\pth}[2][\!]{\mleft({#2}\mright)}%

\newcommand{\ceil}[1]{\left\lceil {#1} \right\rceil}
\newcommand{\floor}[1]{\left\lfloor {#1} \right\rfloor}

\newcommand{\cardin}[1]{\left| {#1} \right|}%

\renewcommand{\th}{th\xspace}

\renewcommand{\Re}{\mathbb{R}}%

\usepackage[inline]{enumitem}

\newlist{compactenumA}{enumerate}{5}%
\setlist[compactenumA]{topsep=0pt,itemsep=-1ex,partopsep=1ex,parsep=1ex,%
   label=(\Alph*)}%

\newlist{compactenuma}{enumerate}{5}%
\setlist[compactenuma]{topsep=0pt,itemsep=-1ex,partopsep=1ex,parsep=1ex,%
   label=(\alph*)}%

\newlist{compactenumI}{enumerate}{5}%
\setlist[compactenumI]{topsep=0pt,itemsep=-1ex,partopsep=1ex,parsep=1ex,%
   label=(\Roman*)}%

\newlist{compactenumi}{enumerate}{5}%
\setlist[compactenumi]{topsep=0pt,itemsep=-1ex,partopsep=1ex,parsep=1ex,%
   label=(\roman*)}%

\newlist{compactitem}{itemize}{5}%
\setlist[compactitem]{topsep=0pt,itemsep=-1ex,partopsep=1ex,parsep=1ex,%
   label=\ensuremath{\bullet}}%

\usepackage{stmaryrd}%
\providecommand{\IntRange}[1]{\mleft\llbracket #1 \mright\rrbracket}
\newcommand{\IRX}[1]{\IntRange{#1}}%

\providecommand{\BibLatexMode}[1]{}
\providecommand{\BibTexMode}[1]{#1}

\ifx\UseBibLatex\undefined%
  \renewcommand{\BibLatexMode}[1]{}
  \renewcommand{\BibTexMode}[1]{#1}
\else
  \renewcommand{\BibLatexMode}[1]{#1}
  \renewcommand{\BibTexMode}[1]{}
\fi

\BibLatexMode{%
   \usepackage[bibencoding=utf8,style=alphabetic,backend=biber]{biblatex}%
   \usepackage{sariel_biblatex}%
}

\providecommand{\Mh}[1]{#1}%

\newcommand{\tldO}{\scalerel*{\widetilde{O}}{j^2}}%

\newcommand{\hp}{\Mh{\mathcalb{g}}}
\newcommand{\hpA}{\Mh{\mathcalb{h}}}
\newcommand{\dualX}[1]{#1^{\star}}%
\newcommand{\dualDualX}[1]{#1^{\star\star}}%

\newcommand{\Line}{\Mh{\mathcalb{l}}}%
\newcommand{\OLine}{\Mh{{\psi}}}%

\newcommand{\OLines}{\Mh{\mathcal{O}}}

\providecommand{\W}{\Mh{\mathcal{W}}}%

\providecommand{\L}{\Mh{L}}%
\renewcommand{\L}{\Mh{L}}%

\newcommand{\R}{\Mh{R}}%

\providecommand{\H}{\Mh{H}}%
\renewcommand{\H}{\Mh{H}}%

\newcommand{\Q}{\Mh{Q}}%
\newcommand{\LinesX}[1]{\Mh{\mathrm{lines}}\pth{#1}}%

\providecommand{\P}{\Mh{P}}%
\renewcommand{\P}{\Mh{P}}%

\newcommand{\Polygon}{\Mh{D}}%

\newcommand{\opt}{\Mh{\mathcalb{o}}}

\newcommand{\VC}{\Term{VC}\xspace}%
\newcommand{\DAG}{\Term{DAG}\xspace}%
\newcommand{\MWU}{\Term{MWU}\xspace}%

\newcommand{\G}{\Mh{G}}

\newcommand{\RangeSet}{{\Mh{\mathcal{R}}}}

\newcommand{\pa}{\Mh{p}}%
\newcommand{\pb}{\Mh{q}}%
\newcommand{\etal}{\textit{et~al.}\xspace}
\newcommand{\eps}{\varepsilon}%

\newcommand{\IFF}{\iff}

\newcommand{\CR}{\Mh{\mathcal{C}}}%
\newcommand{\corrX}[1]{\Mh{\mathrm{corr}}\pth{#1}}%

\newcommand{\Net}{\Mh{\mathcal{N}}}%
\newcommand{\Term}[1]{\textsf{#1}}

\newcommand{\Arr}{\Mh{\EuScript{A}}}%
\newcommand{\ArrX}[1]{\Arr\pth{#1}}%
\newcommand{\GroundSet}{\Mh{\textsf{X}}}
\newcommand{\range}{\Mh{\mathbf{r}}}

\newcommand{\Dim}{\Mh{\delta}}%
\newcommand{\BadProb}{\varphi}

\newcommand{\region}{\Mh{\mathsf{D}}}%
\newcommand{\face}{\Mh{\mathsf{F}}}%
\newcommand{\cell}{\Mh{\mathsf{C}}}%
\newcommand{\body}{\Mh{\mathsf{B}}}%

\newcommand{\numS}{\Mh{k}}
\newcommand{\si}[1]{#1}

\newcommand{\p}{\Mh{p}}

\newcommand{\Inst}{\Mh{\mathcal{I}}}%

\newcommand{\simplex}{\Mh{\nabla}}

\newcommand{\CHX}[1]{\Mh{\mathcal{CH}}\pth{#1}}
\newcommand{\Caratheodory}{Carath\'eodory\xspace}

\newcommand{\fr}{\Mh{\xi}}%

\numberwithin{figure}{section}%
\numberwithin{table}{section}%
\numberwithin{equation}{section}%

\newcommand{\VT}{\Mh{\mathcal{T}}}%
\newcommand{\VTA}{\Mh{\mathcal{U}}}%
\newcommand{\slabX}[1]{\Mh{slab}\pth{#1}}%
\newcommand{\trap}{\Mh{\tau}}%

\newcommand{\LP}{\Term{LP}\xspace}%
\newcommand{\Polygons}{\Mh{\mathcal{D}}}%
\newcommand{\BPolygons}{\Mh{\Polygons_\mathrm{bad}}}%
\newcommand{\BRP}{\Mh{\mathcal{B}}}%
\newcommand{\pg}{\Mh{\sigma}}%

\newcommand{\nL}{\Mh{\mathsf{n}}}
\newcommand{\nP}{\Mh{\mathsf{m}}}

\newcommand{\indexX}[1]{\Mh{\mathrm{\Mh{sep}}}\pth{#1}}%

\DefineNamedColor{named}{OliveGreen} {cmyk}{0.44,0,0.95,0.90}

\newcommand{\valX}[1]{\Mh{v}\pth{#1}}

\BibLatexMode{%
   \bibliography{weak_cuttings}
}

\begin{document}

\title{Halving by a Thousand Cuts or Punctures}

\NotSODAVer{%
   \author{%
      Sariel Har-Peled%
      \SarielThanks{%
         Work on this paper was partially supported by a NSF AF award
         CCF-1907400. %
      }%
      \and%
      \si{Da} Wei Zheng%
      \DavidThanks{}%
   }%
}

\date{\today}
\maketitle

\begin{abstract}
    For point sets $\P_1, \ldots, \P_\numS$, a set of lines $\L$ is
    \emph{halving} if any face of the arrangement $\ArrX{\L}$ contains
    at most $|\P_i|/2$ points of $\P_i$, for all $i$. We study the
    problem of computing a halving set of lines of minimal
    size. Surprisingly, we show a polynomial time algorithm that
    outputs a halving set of size $O(\opt^{3/2})$, where $\opt$ is the
    size of the optimal solution. Our solution relies on solving a new
    variant of the weak $\eps$-net problem for corridors, which we
    believe to be of independent interest.

    We also study other variants of this problem, including an
    alternative setting, where one needs to introduce a set of guards
    (i.e., points), such that no convex set avoiding the guards
    contains more than half the points of each point set.
\end{abstract}

\SODAVer{%
   \thispagestyle{empty}%
   \newpage%
   \setcounter{page}{1}%
}

\section{Introduction}

A basic problem in algorithms is to partition data effectively
in order to apply divide and conquer algorithms, or just
store the data or manipulate it efficiently in distributed or parallel
fashion. In the context of Computational Geometry, such tasks are
usually achieved using cuttings \cite{cf-dvrsi-90}, partitions
\cite{m-ept-92}, or even hashing \cite{ai-nohaa-08}. More recently,
there was significant progress \cite{s-pmit-22} on using polynomials
to perform such partitions (e.g., polynomial ham-sandwich theorem) to
derive better combinatorial bounds (and in some cases, algorithms).

Partitioning a point set $\P \subset \Re^2$ via polynomials is quite
powerful, as such partitions can have many desirable properties not
achievable by the other techniques. However, while computing the
partitioning polynomial can be done efficiently \cite{s-pmit-22},
using such partitions algorithmically is challenging. As a concrete
example, consider a two dimensional polynomial in the plane $p(x,y)$
used to partition a set of points $\P$. It partitions the plane into
cells via its zero set $Z = \Set{(x,y) }{p(x,y)=0}$ -- that is, every
connected component $C$ of $\Re^2 \setminus Z$ induces a cluster in
the partition of $\P$ (i.e., $C \cap \P$). However, computing these
clusters is not algorithmically easy (or convenient) as dealing with
roots of high degree polynomials is cumbersome and computationally
slow. If one remembers how the polynomial $p$ was computed, in some
cases, such tasks become easier -- however, other tasks like adjusting
the partition when the underlying point set changes remains a
challenge, as multi-variable high-degree polynomials are unwieldy.

\subsection{Problem I: Separating multiple point sets by lines/planes}

For a set $\P$ of $n$ points in $\Re^d$, a set $\L$ of (hyper)planes
\emphi{separates} $\P$, if for any pair of points of
$\pa, \pb \in \P$, there is a plane in $\L$ that intersects the
interior of the segment $\pa \pb$ (which also does not contain $\pa$
or $\pb$). In the plane $\L$ is a set of lines.  The
\emphi{separability} of $\P$, denoted by $S_n = \indexX{\P}$, is the
size of the smallest set of lines that separates $\P$. The
separability of a point set captures how grid-like the point set
is. In particular, the separability of the $\sqrt{n}\times\sqrt{n}$
grid is $2\sqrt{n} - 2$, while for $n$ points in convex position the
separability is $\ceil{n/2}$ (and this is the worst case assuming
general position).

This problem can be stated as a hitting set problem (i.e., pick a
minimal size set of planes that hits all the segments formed by pairs
of points). The standard greedy algorithm yields a $O( \log n)$
approximation, and at least in the plane, it can be sped up by using
data-structures \cite{hj-ospl-20}. Somewhat surprisingly, the
separability of random points (picked inside a unit square) is
(roughly) $\Theta(n^{2/3})$, in contrast to grids where it is
$\Theta(\sqrt{n})$ \cite{hj-ospl-20}. Since the separability of $n$
points requires $\Omega(\sqrt{n})$ lines, an approximation quality of
$O( \log n)$ is somewhat more acceptable (although, whether this
approximation ratio can be improved in this case remains open).

\begin{figure}[t]
    \centerline{\includegraphics{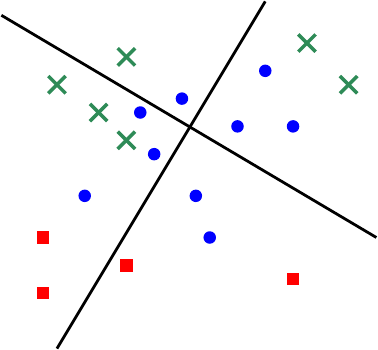}}%
    \caption{Given three point sets, suppose the goal is to break the
       green (cross) point set into sets with at most three points,
       the blue (dot) point set into sets with most four points, and
       the red (square) point set into sets with at most two
       points. This can be achieved using two separating lines.}
    \figlab{example}
\end{figure}

\paragraph*{Separating point sets by lines}

One can define murkier partition problems, such as partitioning
several point sets in a balanced way simultaneously.

\begin{problem}
    \problab{prob}%
    An instance $\Inst$ of the \emphi{reduction} problem is specified
    by $\numS$ point sets $\P_1, \ldots, \P_\numS \subset \Re^2$, not
    necessarily disjoint, and corresponding fractions
    $\fr_1, \ldots, \fr_\numS \in (0,1]$.  The \emphw{size} of $\Inst$
    is $\nP = \sum_i \cardin{\P_i}$.  The goal is to compute the
    smallest set of lines $\L$, such that for every cell $\cell$ in
    the arrangement $\ArrX{\L}$ of lines,
    \begin{math}
        \cardin{\P_i \cap \cell}%
        \leq \fr_i \cardin{\P_i},
    \end{math}
    for all $i$.  See \figref{example} for an example.  In the
    \emphi{halving problem}, we have that $\fr_i = 1/2$.
    Let $\opt$ denote the size of the optimal set $\L$.
\end{problem}
Observe that $\opt$ might be a small constant even if $\numS$ and $n$
are large.

\paragraph{Current solutions to the reduction problem.}
This problem can be reduced to several ``parallel'' instance of
partial set cover \cite{hj-fcmmp-18}, and this in problem can be
stated as a submodular optimization problem, which provides an
$O(\log n)$ approximation by the greedy algorithm. The basic idea is
to define a potential function which captures for every point how far
it is from being happily separated enough from the remaining
points. Then, the greedy algorithm chooses the line that its addition
to the partitioning set reduces this potential function the mos, see
\cite{hj-fcmmp-18} for details.

Using other techniques, Chekuri \etal \cite{ciqvz-acmsc-22} improved
the approximation to $O( \log \numS)$ (assuming that $\fr \geq
1/2$).

\paragraph{The challenge}
In light of the above, the interesting case of the reduction problem
is when the number of sets $\numS$ is polynomially large (e,g.,
$\numS= \sqrt{\nP}$), and the optimal solution $\opt$ is small (say, a
constant). Can one get a constant approximation in this case?

\subsection{Problem II: Guarding multiple point sets against convex
   regions}

The second ``dual'' problem is computing simultaneous weak nets for
several point sets.

\begin{problem}
    \problab{prob:guarding}%
    An instance $\Inst$ of the \emphi{guarding} problem is defined by
    $\numS$ point sets $\P_1, \ldots, \P_\numS \in \Re^d$ (not
    necessarily disjoint), and corresponding fractions
    $\fr_1, \ldots, \fr_\numS \in (0,1]$.  The \emph{size} of $\Inst$
    is $\nP = \sum_i |\P_i|$.  The goal is to find a minimum set of
    points $\Net$, such that every convex region $\region$ where
    $|\region \cap \P_i| \ge \fr_i \cardin{\P_i}$ has nonempty
    intersection with $\Net$.  Let $\opt$ denote the size of the
    optimal set $\Net$.
\end{problem}

This problem can also be viewed as a hitting set problem, where we
wish to hit all convex polygons that contain too many points of a
given point set $\P_i$ by at least one point.

\subsection{Background}
\paragraph*{Weak $\eps$-nets and guarding points from convex regions}
For a set $\P$ of $n$ points in $\Re^d$, a set $S \subset \Re^d$ is a
\emphi{weak $\eps$-net} if for every convex region $\region$ where
$|\region \cap P| \ge \eps n$ has nonempty intersection with $S$. We
can view $S$ as a set of guards in the plane that protects the points
against any convex set which contains many points.  The goal is to
pick a discrete point set where $|S|$ is as small as possible.  This
problem is well studied, see \cite{bfl-nhp-90,abfk-pswnch-92,
   ceggsw-ibwen-95,mw-ncwen-04}.  The state of the art is the recent
results by Rubin \cite{r-ibwen-18,r-sbwenhd-21} showing the existence
of weak $\eps$-nets of size $O_d(\eps^{-(d - 0.5 + \alpha)})$ for
arbitrarily small $\alpha > 0$. For more detailed history of the
problem, see the introduction of Rubin
\cite{r-ibwen-18,r-sbwenhd-21}. As for a lower bound, Bukh \etal
\cite{bmn-lbwensc-09} gave constructions of point sets for which any
weak $\eps$-net must have size
$\Omega_d(\eps^{-1} \log^{d-1} \eps^{-1})$.  Closing this gap remains
a major open problem.  See \cite{mv-eaen-17} for a recent survey of
$\eps$-nets and related concepts.

\bigskip%
\noindent%
\textbf{Round-and-cut.}\footnote{Not to be confused with
   \emph{cut-and-run} or probably the more correct name for this
   technique ``\emph{cut as long as you can not round}''.} %
Many approximation algorithm works by rounding a fractional optimal
solution to an associated \LP. Sometime the \LP is implicit, and it
can be solved using the ellipsoid algorithm via a separation
oracle. In particular, it is well known that \LP can be solved in
(weakly) polynomial time by such an algorithm. At every step, the \LP
solver asks the separation oracle about the status of a specific
solution/point. The oracle either finds a violated constraint and
returns it, or alternatively returns that the given query point is
feasible. Once a solution to the \LP is found, the approximation
algorithm rounds the \LP to get (hopefully) a good approximation.

In the \emphi{round-and-cut} approach \cite{cflp-sigcn-00}, one
combines the two steps. Given a query point (i.e., a fractional
assignment), the oracle either returns a violated constraint (if one
such constraint is easy to find), or tries to round this fractional
value. If the rounding is successful -- a good approximation was
found. Otherwise, the failure of the rounding provides a violated
constraint which is returned by the separation oracle. This is
especially useful where we do not know how to implement the standard
separation oracle, or the standard separation oracle requires
exponential time.

A variant of the round-and-cut technique was used (implicitly) in
computational geometry. The multiplicative weight update (\MWU)
algorithms can be viewed as solving an \LP. In particular, Clarkson's
algorithm \cite{c-apca-93} for set cover/hitting set (see also
\cite{bg-aoscf-95}), work by assigning weights to points (i.e., think
about these as the \LP values assigned to the points), and each stage
computing an $\eps$-net (the value of $\eps$ is guessed in
advance). Either a small $\eps$-net is found, or alternatively a
multiplicative weight update is applied (i.e., the values of the \LP
are adjusted). This connection between these \MWU algorithms and \LP
is discussed by Har-Peled \cite[Chapter 6]{h-gaa-11}.

\subsection{Our results}

We provide polynomial time approximation algorithms for both
problems. The reduction problem is solved by studying a fractional
version of the line separation problem. We solve the later problem
using the round-and-cut framework, where the rounding procedure
(essentially) requires a solution to a new problem, which is
``dual'' to the weak $\eps$-net problem.
Specifically, given a set of lines $\L$,
one need to find a minimum number of lines that intersect all convex
regions intersecting more than $\eps$-fraction of the lines of
$\L$. We refer to this problem as the \emphi{weak $\eps$-cutting}
problem.  This problem has similar flavor to the weak $\eps$-net
problem.

Surprisingly, unlike for the weak $\eps$-net problem, this problem has
a direct $O(1/\eps^2)$ solution.  Indeed, setting $r=1/\eps$, one can
compute a $1/r$-cutting of $\L$. This decomposes $\Re^d$ into $O(r^d)$
simplices, so that each one intersects at most $\nL/r$ lines (i.e.,
planes or hyperplanes in higher dimensions), where $\nL = |\L|$. In
2d, adding the lines supporting the edges of the triangles readily
yields a weak $\eps$-cutting with $O(r^2)$ lines.  This solution also
works in higher dimensions, yielding a weak $\eps$-cutting of size
$O(1/\eps^d)$ in $d$ dimensions.

\paragraph{Smaller weak cuttings in 2d}

In two dimensions we show how to efficiently construct weak
$\eps$-cutting of size $\tldO(1/\eps^{3/2})$.  The construction
requires over-sampling together with a refinement of larger faces into
``large'' polygons, and using known combinatorial bounds on the
complexity of many faces.  See \thmref{w:cutting:2:d} for details.

\paragraph{Weak $\eps$-net for corridors}

In the dual, the above problem becomes the following -- given a set
$\P$ of $\nP$ points in the plane, compute a set $\Net$ of points (not
in $\P$), such that any corridor containing $\eps \nP$ points of $\P$
must contain a point of $\Net$, where a \emph{corridor} is the region
bound between the upper and lower envelopes of any set of lines. That
is, this is the problem of computing \emphw{weak $\eps$-net for
   corridors}. The above constructions readily implies a weak
$\eps$-net of size $\tldO(1/\eps^{3/2})$.

These two problems were not studied before, and we consider this
result (and its primal) to be quite surprising.

\paragraph{Approximation to the reduction problem}

We transform the reduction problem (using lines for separation)
to an implicit hitting set problem. We solve the \LP relaxation of the
later problem, by repeatedly using the weak cutting construction
algorithm above to perform rounding, and find a violated constraint if
such a constraint exists. This replaces the original constraints
involving multiple sets, into ``monochromatic'' constraints. This
yields a set of $O(\opt^{3/2} \log^{3/2} \opt)$ lines that performs
the desired separation, where $\opt$ is the size of the optimal
solution, see \thmref{red:approx} for details. Thus, when $\opt$ is a
constant, our algorithm is the first constant approximation algorithm
for this problem.

Interestingly, to get a fast algorithm, we show that one can reduce
the number of lines under consideration. In particular, we show that
instead of the $O(\nP^2)$ lines, one can quickly generate a candidate
set of lines of size $\tldO\bigl( (\numS/\fr)^2 \bigr)$, such that it
contains a constant approximation to the optimal solution, see
\lemref{lines:r} for details (here $\fr$ is the minimum fraction of
separation required of any set in the original instance). Using this
as a preprocessing stage, yields a near linear time approximation
algorithm for the reduction problem.

\paragraph{Approximation algorithm for the convex guarding problem}

The same approach works for the convex guarding problem, except that
the rounding now is done via the ``standard'' weak $\eps$-net
construction. Furthermore, the separation oracle requires finding a
bad convex polygon given a suggested net, which is done via dynamic
programming, which might be of independent interest (see
\lemref{bad:polygon}). In this case, we do not have a way to reduce
the candidate set of points being used as part of the net, and thus
the running time is worse (i.e., polynomial). We get a
$\tldO(\sqrt{\opt})$ approximation in polynomial time, see
\thmref{g:convex:polygons}
for details.

\paragraph*{Paper organization}

We start at \secref{prelims} with some standard background.
In \secref{weak:e:net} we present the construction of weak
$\eps$-cutting for lines, which in the dual is weak $\eps$-net for
corridors. We present the approximation algorithm for the reduction
problem in \secref{red:problem}.  In \secref{g:convex} we present the
approximation algorithm for the convex guarding problem.

\section{Preliminaries}
\seclab{prelims}

\subsection{Notations}

For an integer $n$, let $\IRX{n} = \{1, \ldots, n\}$. In the
following, \emphi{plane} denotes a flat of dimension $d-1$ contained
in $\Re^d$.

\paragraph{Duality}
A plane $\hp \equiv x_d = b_1 x_1 + \cdots + b_{d-1} x_{d-1} + b_d$ in
$\Re^d$ can be interpreted as a function from $\Re^{d-1}$ to
$\Re$. Given a point $\p = (\p_1, \ldots, \p_d)$, let
$\hp(\p) = b_1 \p_1 + \cdots + b_{d-1} \p_{d-1} + b_d$. Thus, a point
$\p$ lies \emph{above} the plane $\hp$ if $\p_d > \hp( \p )$.  As
such, a point lies \emph{on} the plane $\hp$ if $\hp(\p) = \p_d$. The
\emphi{duality} between points and planes is defined as
\begin{equation*}
    \begin{aligned}
  \p= (\p_1,\ldots, \p_d)%
  &\quad \implies \quad%
    \dualX{\p} \equiv x_d = \p_1 x_1 + \cdots+
    \p_{d-1}x_{d-1}
    - \p_d\\
  \hp \equiv x_d = a_1 x_1 + \cdots +
  a_{d-1}x_{d-1} + a_d %
  &\quad \implies \quad%
    \dualX{\hp} = (a_1, \ldots, a_{d-1}, -a_d).
\end{aligned}
\end{equation*}

The following is well known \cite{h-gaa-11}.
\begin{lemma}
    \lemlab{duality:extent}%
    For a point $\p = (b_1, \ldots, b_d)$, we have the following:
    \begin{compactenumA}
        \item $\dualDualX{\p} = \p$.

        \item A point $\p$ lies above/below/on the plane $\hp$
        $\Longleftrightarrow$ %
        \NotSODAVer{the } point $\dualX{\hp}$ lies above/below/on the
        plane $\dualX{\p}$.

        \item The vertical distance between $\p$ and $\hp$ is the same
        as that between $\dualX{\p}$ and $\dualX{\hp}$.

        \item The distance between two parallel planes $\hp$ and
        $\hpA$ is the length of the vertical segment
        $\dualX{\hp}\dualX{\hpA}$.

    \end{compactenumA}
\end{lemma}

\subsection{The reduction problem as a hitting set problem }

\probref{prob} specifies the given instance.  A convex region
$\region$, is \emphi{bad} for a set of lines $\L$ if $\region$ does
not intersect any line of $\L$, and there is an index $i$, such that
$|\P_i \cap \region| > \fr_i |\P_i|$.  In particular, \probref{prob}
can be interpreted as computing a minimal set of lines that intersects
the interior of all the bad regions. That is, this problem can be
stated as a hitting set problem. Since the family of all convex
regions induced by a set of all allowable lines $\L'$ has an
exponential size in $\cardin{\L'}$, it is not possible to compute this
family of regions explicitly in polynomial time or explicitly state
the \LP associated with this problem (since it has exponential size).

\subsection{Corridors}

Given a plane (i.e., a line in two dimensions)
\begin{math}
    \hp \equiv x_d = \sum_{i=1}^{d-1} \lambda_i x_i + \lambda_d,
\end{math}
and a non-zero real number $\beta \in \Re$, let $\alpha \otimes \hp$
be the \emphi{scaled} plane
\begin{math}
    \alpha \otimes \hp \equiv x_d = \sum_{i=1}^{d-1} \alpha \lambda_i
    x_i + \alpha \lambda_d,
\end{math}
Similarly, given a second plane
\begin{math}
    \hpA \equiv x_d = \sum_{i=1}^{d-1} \lambda'_i x_i + \lambda'_d,
\end{math}
let their \emphi{sum} be the plane
\begin{math}
    \hp \oplus \hpA \equiv x_d = \sum_{i=1}^{d-1} (\lambda_i +
    \lambda'_i) x_i + (\lambda_d + \lambda'_d),
\end{math}
Thus, planes form a vector space with scalar multiplication $\otimes$,
and vector addition $\oplus$ (this is an immediate consequence of
``importing'' the corresponding operations from the dual space).

\newcommand{\hullX}[1]{\Mh{\mathcal{H}}\pth{#1}}

Thus, given two planes $\hp_1, \hp_2$, their \emphi{convex
   combination}, for $t \in [0,1]$, is the plane
\begin{equation*}
    \hp(t) = \bigl((1-t) \otimes \hp_1\bigr) \oplus \bigl(t\otimes \hp_2\bigr).
\end{equation*}
For two lines in the plane, the set $\Set{\hp(t)}{t\in [0,1]}$ is the
set of all lines in the double wedge between $\hp_1$ and $\hp_2$
passing through the intersection point $\hp_1 \cap \hp_2$.  More
generally, given $\alpha_1, \ldots, \alpha_m \in [0,1]$, with
$\sum_i \alpha_i =1$, and planes $\hp_1, \ldots, \hp_m$, they define
the \emphi{convex combination}
$(\alpha_1 \otimes \hp_1) \oplus\cdots \oplus (\alpha_m \otimes
\hp_m)$. Given a set of planes $\L$, their \emphi{hull}, denoted by
$\hullX{\L}$, is the set of all their convex-combinations.

A set of planes $\L$ is \emphi{convex} if $\L = \hullX{\L}$.  For a
set of planes $\L$, its \emphi{corridor}
$\corrX{\L} = \cup\hullX{\L} = \displaystyle\cup^{}_{\hp \in
   \hullX{\L}}\, \hp$ is the union of planes in $\hullX{\L}$ --
geometrically, it is the region bounded by the upper and lower
envelopes of $\L$. A set of planes $\L$ is \emphi{bounded}, if all the
coefficients used by planes in $\L$ are bounded. The closure of a
bounded convex set of planes does not contain vertical planes. A
region $R \subseteq \Re^d$ is a \emphi{corridor} if there is a set of
planes $\L$, such that $R = \corrX{\L}$.

\begin{lemma}
    \lemlab{corridor:w}%
    Let $\L$ be a finite (bounded) set of planes in $\Re^d$.  We have
    the following:
    \begin{compactenumA}
        \item The dual of $\CR = \corrX{\L}$ is a convex polytope
        $\dualX{\CR} = \dualX{\pth{\corrX{\L}}} = \CHX{\dualX{\L}}$.

        \item A point $\pa \in \corrX{\L}$ $\IFF$ the plane
        $\dualX{\pa}$ intersects $\CHX{\dualX{\L}}$.

        \item A plane $\hp \subseteq \corrX{\L}$ $\IFF$ the point
        $\dualX{\hp} \in \CHX{\dualX{\L}}$.

        \item A plane $\hp \subseteq \corrX{\L}$ $\implies$
        $\exists \H = \{ \hpA_1, \ldots, \hpA_{d+1 } \} \subseteq \L$,
        such that $\hp \subseteq \corrX{ \H }$. Furthermore, the set
        $\H$ can be computed in $O(n)$ time.
    \end{compactenumA}
\end{lemma}
\begin{proof}
    (A), (B) and (C) are immediate implications of duality.

    (D) Indeed, we have $\hp \subseteq \corrX{\L}$ $\IFF$
    $\dualX{\hp}\in \CHX{\dualX{\L}}$. By \Caratheodory's theorem,
    there is a set of $d+1$ points
    $\dualX{\H} = \{ \dualX{\hpA_1}, \ldots, \dualX{\hpA_{d+1}} \}
    \subseteq \dualX{\L}$ such that
    $\dualX{\hp} \in \CHX{\dualX{\H}}$. This set can be computed in
    linear time using low dimensional linear programming
    \cite{h-gaa-11}. As such, for $\H = \dualDualX{\H}$, we have that
    $\hp = \dualDualX{\hp} \subseteq \corrX{\H}$.
\end{proof}

\begin{defn}
    For a set $\P$ of $\nP$ points in the plane, and a parameter
    $\eps \in (0,1)$, a \emphi{weak $\eps$-net for corridors} is a set
    of points $\Net \subseteq \Re^2$, such that for any corridor $\CR$
    in the plane, that contains at least $\eps \nP$ points of $\P$,
    must contain at least one point of $\Net$.
\end{defn}

\begin{figure}[t]
    \centerline{\includegraphics{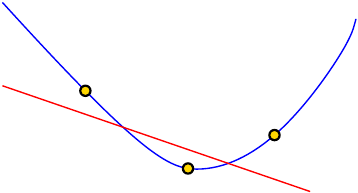}}%
    \caption{%
       Including a point in the corridor by choosing a line separating
       it from remaining points.  }
     \figlab{shattering_corridor}
\end{figure}

The following implies that ``strong'' versions of weak $\eps$-nets for
corridors, where $\Net$ is restricted to be a subset of $\P$, must
contain a large number of points.

\begin{lemma}
    Let $\eps \in (0,1)$, and let $\nP$ be any positive integer number.
    There is a set $\P$ of $\nP$ points in the plane, such that any
    weak $\eps$-net for corridors that is restricted to be a subset of
    $\P$ must be of size $\geq \nP-\ceil{\eps \nP} +1$.
\end{lemma}

\begin{proof}
    Let $f(x) = x^2$ (or any other convex smooth function).
    Let $\P = \Set{ (i,f(i)) }{ i \in \IRX{\nP}\bigr.}$,
    and let $\ell$ be any line lying below $\P$.

    Consider any set $X \subseteq \P$.  For any point $\p \in X$,
    introduce a non-vertical line lying above $p$ that seperates $\p$
    from the remaining points of $\P$. Clearly, for the resulting set of lines $\L_X$, we
    have $\corrX{ \L_X \cup \{\ell\}} \cap \P = X$. Thus, any $\eps$-net
    $Y \subseteq \P$, such that $|Y| < \nP - \ceil{\eps \nP} +1$ does
    not stab the complement set $X = \P \setminus X$, which is both
    $\eps$-heavy and realizable as a corridor, as
    $\corrX{\L_X \cup \{\ell\}} \cap \P = X$, a contradiction.
\end{proof}

\subsection{$\eps$-net theorem}

\begin{defn}
    A set $\Net \subseteq \GroundSet$ is an \emphi{$\eps$-net} for
    $\GroundSet$ if for any range $\range \in \RangeSet$, if
    $\cardin{\range \cap \GroundSet} \geq \eps \cardin{\GroundSet}$,
    then $\range$ contains at least one point of $\Net$ (i.e.,
    $\range \cap \Net \ne \emptyset$).
\end{defn}

\begin{theorem}[$\eps$-net theorem, \cite{hw-ensrq-87}]
    \thmlab{epsilon:net}%
    Let $(\GroundSet,\RangeSet)$ be a range space of \VC dimension
    $\Dim$, and suppose that $0 < \eps \leq 1$ and $\BadProb < 1$. Let
    $\Net$ be a set obtained by
    \begin{math}
        m%
        =%
        \Omega\pth{ \eps^{-1} (\log \BadProb^{-1} + \Dim \log
           \eps^{-1} )\bigr.}
    \end{math}
    random independent draws from $\GroundSet$.  Then $\Net$ is an
    $\eps$-net for $\GroundSet$ with probability at least
    $1-\BadProb$.
\end{theorem}

\section{Weak $\eps$-net for corridors}
\seclab{weak:e:net}

The input is a set $\L$ of $\nL$ lines in the plane, and a parameter
$\eps \in (0,1)$. Our purpose here is to compute a set $\Net$ of lines
in the plane, such that any convex body $\body$ in the plane such that
$|\body \sqcap \L| \geq \eps \nL$, we have that a line of $\Net$
intersects $\body$, where
\begin{equation*}
    \body \sqcap \L
    =%
    \Set{ \Line \in \L}{ \Line \cap \body \neq \emptyset}.
\end{equation*}
In the dual, this corresponds to the property that any corridor $\CR$
containing more than $\eps \nL$ points of $\dualX{\L}$, contains a
point of $\dualX{\Net}$. The set $\Net$ is a \emphi{weak
   $\eps$-cutting}, as every face of the arrangement $\ArrX{\Net}$
intersects at most $\eps \nL$ lines of $\L$. This definition extends
naturally to higher dimensions.

We start with an easy construction.

\begin{lemma}
    For a set $\L$ of $\nL$ planes in $\Re^d$, one can compute a weak
    $\eps$-cuttings size $O(1/\eps^d)$.
\end{lemma}
\begin{proof}
    Compute regular cuttings of size $O(r^d)$ of $\L$, for
    $r = \ceil{1/\eps}$.
    Furthermore, such cuttings decompose $\Re^d$ into
    $O(r^d)$ simplices, where each simplex intersects at most $\nL/r$
    planes of $\L$.  We replace each $(d-1)$-dimensional face of a
    simplex in the cutting by the plane that supports it. Clearly, the
    resulting set $\Net$ of planes is of size $O(r^d)$, and fulfils the
    requirement of being a weak $\eps$-cutting.

    Indeed, consider any convex region $\body \subseteq \Re^d$. If
    $\body$ is not fully contained in a simplex of the cuttings, then
    it must intersect one of the planes of $\Net$. Otherwise, it is
    contained in a simple simplex of the cutting, say $\simplex$. But
    then, we have
    \begin{math}
        \cardin{\body \sqcap \L} \leq \cardin{\simplex \sqcap \L} \leq
        \eps \nL.
    \end{math}
\end{proof}

\subsection{A better construction in two dimensions}

A better construction, in two dimensions, arises by oversampling
coupled with breaking down the large faces into polygons with fewer
edges.

\paragraph{Complexity of the $i$\th largest face}

Let $\L$ be a set of lines in the plane, and consider the arrangement
$\Arr = \ArrX{\L}$.  The \emphi{complexity} of a face $\face$ of
$\Arr$, denoted by $|\face|$, is the number of edges/rays on the
boundary of $\face$.

\begin{lemma}
    \lemlab{dec:complexity}%
    For a set $\L$ of $\nu$ lines in the plane, let $c_i$ be the
    complexity of the $i$\th face of $\Arr=\ArrX{\L}$ in decreasing
    order of the complexity of the faces. Then
    \begin{math}
        c_i = O( \nu^{2/3} /i^{1/3} + \nu/i + 1).
    \end{math}
\end{lemma}
\begin{proof}
    Let $c_i$ be the complexity of the largest $i$\th face in
    $\Arr$. The complexity of $i$ faces in the arrangement of $\Arr$
    is $M_i =  O( \nu^{2/3}i^{2/3} + \nu + i)$ \cite{sa-dsstg-95}.
    Namely, we have
    \begin{math}
        i c_i%
        \leq %
        \sum_{j=1}^i c_i%
        \leq%
        M_i,
    \end{math}
    which implies
    \begin{math}
        c_i \leq M_i /i.
    \end{math}~
\end{proof}

\begin{theorem}
    \thmlab{w:cutting:2:d}%
    Let $\L$ be a set of $\nL$ lines in $\Re^2$, and let
    $\eps \in (0,1)$ be a parameter. One can compute a set $\R$ of
    lines of size $O(\eps^{-3/2} \log^{3/2} \eps^{-1} )$, such that
    $\R$ is a weak $\eps$-cuttings of $\L$. That is, any open convex
    region $\region$ that avoids the lines of $\R$, intersects at most
    $\eps \nL$ lines of $\L$.

    In the dual, $\dualX{\R}$ is a weak $\eps$-net for corridors for
    the point set $\dualX{\L}$. That is, any corridor $\CR$ that
    avoids the points of $\dualX{R}$ contains at most $\eps \nL$
    points of $\dualX{\L}$.
\end{theorem}
\begin{proof}
    Let $r = \ceil{10/\eps}$. Let $\R_1$ be a random sample from $\L$
    of size $\nu = c \alpha r \log r$, where $c$ is a sufficiently
    large constant, and $\alpha \in \IRX{r^3}$ is a parameter. The
    sample $\R_1$ is an $\delta$-net for $\L$ for vertical trapezoids,
    where $\delta = 1/(2\alpha r)$, with probability close to one. In
    particular, consider a face $\face$ of $\ArrX{\R_1}$. If $\face$
    has at most $\alpha$ edges, then it can be decomposed into
    $\alpha$ vertical trapezoids, each one intersecting at most
    $\delta \nL$ lines of $\L$. As such,
    $|\face \cap \L| \leq \alpha \delta \nL \leq \nL/r $,

    Thus, we need to fix only large faces with strictly more than
    $\alpha$ edges. Let $\face$ be such a face, and sweep it from left
    to right by a vertical line, whenever the sweep line encounters
    the $\alpha i$\th vertex of $\face$, introduce a vertical line to
    break it into smaller faces. Let $\R_2$ be the resulting set of
    new lines introduced. If the total number of edges of faces with
    more $\alpha$ edges is $I$, then the overall number of lines
    introduced is $I/\alpha$. Let $\R = \R_1 \cup \R_2$.

    Consider a face $\face$ of $\ArrX{\R}$ that is contained in a face
    $\face'$ of $\ArrX{\R_1}$. There are several possibilities:
    \begin{compactenumi}
        \smallskip
        \item $\face$ has at most $\alpha$ edges. Then, $\face$ can be
        decomposed into $\alpha$ vertical trapezoids that avoids
        $\R_1$, and $\cardin{\face \cap \L} \leq \nL/r$, following the
        argument above.

        \smallskip
        \item The face $\face'$ has at most $\alpha$ edges -- the same
        argument implies
        $\cardin{\face \cap \L} \leq \cardin{\face' \cap \L} \leq
        \nL/r$.

        \smallskip
        \item Otherwise, the above process introduced vertical lines
        into $\R_2$ that break $\face'$ into polygons with at most
        $\alpha$ edges. In particular, one of these polygons, say,
        $\Polygon \subseteq \face'$ contains $\face$. Arguing as
        above, we have
        $\cardin{\face \cap \L} \leq \cardin{\Polygon \cap \L } \leq
        \nL/r$.
    \end{compactenumi}
    \smallskip%
    Now, consider any convex region $\body$ that avoids the lines of
    $\R$, and observe that it is contained in a single face of
    $\ArrX{\R}$, which intersects at most $\eps \nL$ lines of $\L$, by
    the above. Thus $\R$ is the desired weak $\eps$-cutting.

    Recall that $\nu = \cardin{\R_1} = c \alpha r \log r$. By
    \lemref{dec:complexity}, if we require that
    $ \nu^{2/3} /m^{1/3} + \nu/m + 1 \leq c'\alpha$, for $c'$ a
    sufficiently small constant. Thus, the complexity of the $m$\th
    face in $\ArrX{\R_1} $ is at most $\alpha$. This holds if
    $\nu^{2/3} /m^{1/3}\leq c'\alpha/3$ and $\nu /m\leq
    c'\alpha/3$. This in turns holds if \NotSODAVer{%
       \begin{equation*}
           m \geq  \pth{\frac{\nu^{2/3}}{c' \alpha/3}}^3
           =%
           \Omega\pth{ \frac{\nu^2}{\alpha^3} }
           =%
           \Omega\pth{ \frac{ r^2 \log^2 r}{\alpha} }.
       \end{equation*}%
    }
    \SODAVer{
    \begin{math}
        m \geq  \pth{\frac{\nu^{2/3}}{c' \alpha/3}}^3
        =%
        \Omega\pth{ \frac{\nu^2}{\alpha^3} }
        =%
        \Omega\pth{ \frac{ r^2 \log^2 r}{\alpha} }.
    \end{math}%
 } The total complexity of these $m$ large faces is
        \begin{math}
            I%
            =%
            O(m^{2/3} \nu^{2/3} + m + \nu)%
            =%
            \NotSODAVer{%
               O\Bigl( \pth{ \frac{ r^2 \log^2 r}{\alpha} }^{2/3}
               (\alpha r \log r)^{2/3} \Bigr)%
               =}%
            O( r^2 \log^2 r).
        \end{math}
        As such, the set $\R$ has size
        $\nu + I /\alpha = O\pth{ \alpha r \log r + \frac{r^2\log^2
              r}{\alpha}}$, which is minimized for
        $\alpha = O( \sqrt{r \log r})$.~
\end{proof}

\section{The reduction problem:  Approximation algorithm}
\seclab{red:problem}

The input instance $\Inst$ is made of $\numS$ point sets
$\P_1, \ldots, \P_\numS$ in $\Re^2$, not necessarily disjoint, and
fractions $\fr_1, \ldots, \fr_\numS \in (0,1]$. Let
$\fr = \min_i \fr_i$. Furthermore, let $\nP_i = \cardin{\P_i}$, for
$i \in \IRX{\numS}$, and $\nP = \sum_i \cardin{\P_i} $.  As a
reminder, the goal is to compute the smallest set of lines $\L$, such
that for every cell $\cell$ in the arrangement $\ArrX{\L}$ of lines,
\begin{math}
    \cardin{\P_i \cap \cell}%
    \leq%
    \fr_i \cardin{\P_i},
\end{math}
for all $i$.

\subsection{Reducing the number of candidate cutting lines}

\begin{observation}
    Consider the range space where the ground set is $\Re^2$, and
    ranges are corridors formed by $3$ lines. The \VC dimension of
    this range space is $O(1)$.
\end{observation}

Given a set $\Q$ of points in general position. Let $\LinesX{\Q}$ be
the set of all lines passing through pairs of points of $\Q$.  The set
$\LinesX{\Q}$ has size $O( |\Q|^2)$, and it can be computed in this
time.

\begin{lemma}
    \lemlab{lines:r}
    Given an instance $\Inst$ to the reduction problem, in the plane,
    with $\numS$ different sets and $\fr = \min_i \fr_i$. One can
    compute a set of lines $\L$, of size
    \begin{math}
        O( (\numS / \fr)^2 \log^2 (\numS / \fr) ),
    \end{math}
    such that there is a solution for $\Inst$ of size $\leq 3\opt$,
    made out of lines from $\L$, where $\opt$ is the size of the
    optimal solution (where any line in he plane can be used). The
    running time of the algorithm is bounded by the output size.
\end{lemma}

\begin{proof}
    Let $\Net_i$ be a $\fr_i/2$-net of $\P_i$ for corridors formed by
    three lines. By the above, a sample of size
    $O( \fr_i^{-1} \log (\numS/\fr_i) )$ is such a net with
    probability $\geq 1 - 1/\numS^{O(1)}$. Let $\Net = \cup_i \Net_i$,
    and let $\L = \LinesX{\Net}$.

    We claim that $\L$ is the desired set of lines. Assume there is an
    optimal solution $\OLines = \{ \OLine_1, \ldots \OLine_\opt
    \}$. Each line $\OLine_i$, separates $\Net$ into two sets
    $\Q_i^+, \Q_i^-$. Consider the polygons $\CHX{\Q_i^-}$ and
    $\CHX{\Q_i^+}$. Let $\L_i \subseteq \L$ be the lines the edges of
    these convex-hulls induces, as well as the two lines realizing the
    cross tangents.

    The corridor $\corrX{\L_i}$ does not contain any point of $\Net$
    in its interior, and $\OLine_i \subseteq \corrX{\L_i}$.  By dual
    \Caratheodory theorem, \lemref{corridor:w} (D), there is a set
    $\L_i' \subseteq \L_i$ of three lines, such that
    $\OLine_i \subseteq \corrX{\L_i'}$.  We claim that
    $\H = \cup_i \L_i'$ is a valid solution to the reduction set, and
    $|\H| \leq 3 \opt$.

    Consider any face $\face$ of $\ArrX{\H}$. If it is contained in a
    face $\face'$ of $\ArrX{\OLines}$, then, for any $i$, we have
    $|\face \cap \P_i| \leq |\face' \cap \P_i| \leq \fr_i
    |\P_i|$. Otherwise, $\face$ must be crossed by an optimal line,
    say, $\OLine_i$. But then, $\face \subseteq \corrX{\L_i'}$. This
    implies that, for any $j$, we have
    \begin{equation*}
        |\face \cap \P_j|
        \leq%
        |\corrX{\L_i'} \cap \P_j|
        \leq%
        (\fr_j/2) |\P_j|,
    \end{equation*}
    since $\corrX{\L_i'}$ does not contain any point of $\Net$ (and
    thus of $\Net_j$) in its interior, and $\Net_j$ is a $\fr_j/2$-net
    for corridors induced by three lines for $\P_j$. In the above, we
    treated both $\face$ and $\corrX{\L_i'}$ as open sets.  One need
    to repeat the above argument also for edges of the arrangement
    $\ArrX{\H}$, but this case is easier, as can be easily verified.
    We have
    $|\Net | = \sum_i O( \fr_i^{-1} \log (\numS/\fr_i) ) = O\bigl(
    (\numS/\fr) \log (\numS/\fr ) \bigr)$, and thus
    $|\L| = O(|\Net|^2 ) = O\bigl( (\numS/\fr)^2 \log^2 (\numS/\fr )
    \bigr)$.~
\end{proof}

\subsection{Implicit \LP, separation oracle and fractional solution}

Given an instance $\Inst$ of the reduction problem, a set $\L$ of
$\nL$ lines, consider the set of all ``bad'' regions. To this end, let
$\Polygons$ be the set of all convex sets in the plane.  We assume
that no two vertices of $\ArrX{\L}$ have the same $x$ value.

A convex region $\pg \in \Polygons$ is \emphi{bad} if there is an index $j$,
such that $|\P_j \cap \pg | > \fr_j |\P_j|$, and let $\BPolygons$ be
the set of all bad regions in $\Polygons$. The associated \LP for
computing a hitting set, of at most $t$ lines, for all bad polygons
is:
\begin{align}
  \valX{L} =
  &
    \sum_{\Line \in \L} x_\Line \leq t
    \eqlab{lp:t}
  \\
  &1 \geq x_\Line \geq 0
  & \forall \Line \in \L
    \nonumber
  \\
  \valX{\pg} = &\sum_{\Line \in \L \sqcap \pg} x_\Line \geq 1
  & \forall \pg \in \BPolygons. \tag{*}
\end{align}
A \emphi{separation oracle} is a procedure that gets assignment of
fractional values for the variables of the \LP, and returns a violated
constraint if such a constraint exists. For our purposes it is enough
to find an approximate violation.  The first two conditions in the
above \LP can be checked directly, and if they are violated, they are
returned as violated to the \LP solver. Otherwise, the algorithm tries
to round the solution as described next.

\subsubsection{The rounding attempt}

A fractional solution to the above \LP can be efficiently rounded. One
need to adapt the algorithm of \secref{weak:e:net} to work for the
case that the lines have weights. Conceptually, we treat
$\alpha =v(\L)$ as the number of lines we have, and
$\eps = 1/(2\alpha)$ as the desired threshold. The $\eps$-net theorem
applies verbatim in the weighted settings (the sampling has to be
adapted to the weights, but this is standard), and the algorithm of
\secref{weak:e:net} applies verbatim. We get a weak $\eps$-cutting of
size $O(\alpha^{3/2} \log^{3/2} \alpha )$, realized by a set $\W$ of
lines. For every face $\face$ of $\ArrX{\W}$ the algorithm computes
\begin{equation*}
    \valX{\face} = \sum\nolimits_{\Line \in \L \sqcap \face} v(\Line),
    \qquad\text{and}\qquad
    |\face \cap \P_j|, \quad j=1,\ldots, k.
\end{equation*}
If $ \valX{\face} \leq 1/2$ then the weak cutting computed failed
(which happens with low probability), and the algorithm recomputes the
weak cutting.  If there is any $j$ such that
$|\face \cap \P_j| > \fr_j |\P_j|$, then the rounding
failed. Namely, we found a bad region $\pg \in \BPolygons$ (i.e., a
constraint of the \LP that is violated). The algorithm returns the
corresponding constraint of (*) as being violated
\begin{equation*}
    \valX{\face} = \sum\nolimits_{\Line \in \L \sqcap \face} x_\Line \geq 1,
\end{equation*}
as $\valX{\face} < 1/2$.  If all the faces of $\ArrX{\W}$ are good,
then $\W$ is a valid solution, with
$O( \alpha^{3/2} \log^{3/2} \alpha)$ lines.

\begin{remark}
    The ellipsoid algorithm (with a separation oracle) solves an \LP
    with $\nL$ variables, using a number of iterations that is
    polynomial in $\nL$ and $\log$ of the largest number in the \LP,
    which is $\nL$ (in our case). See \cite{gls-gaco-93}.
\end{remark}

\begin{theorem}
    \thmlab{red:approx}%
    Given an instance
    $\Inst =(\P_1,\fr_1, \ldots, \P_\numS, \fr_\numS)$ of the
    \probrefY{prob}{reduction problem} in the plane of size $\nP$,
    with $\numS$ sets and $\fr = \min_i \fr_i$, one can compute a set
    $\L$ of $O(\opt^{3/2} \log^{3/2} \opt)$ lines, such that for any
    cell $\cell$ of $\ArrX{\L}$, and any $j\in\IRX{\numS}$, we have
    that $|\P_j \cap \cell| \leq \fr_j |\P_j|$, where $\opt$ is the
    minimum size of any set of lines with this property.

    The expected running time of this algorithm is
    $O\bigl( \nP \log( \numS / \fr) + ( \numS / \fr )^{O(1)}
    \bigr)$, and the algorithm succeeds with probability
    $\geq 1- (\fr/\numS)^{O(1)}$.
\end{theorem}
\begin{proof}
    A naive upper bound on $\opt$ is $O(\numS/\fr)$, as each set
    $\P_i$ can be partitioned using $\ceil{1/\fr_i}$ lines.

    We first generate a small set of lines $\L$ as candidates for
    cutting the set, using \lemref{lines:r}.  This stage succeeds with
    probability $\geq 1- (\fr/\numS)^{O(1)}$.  The resulting set of
    lines $\L$ has size
    \begin{math}
        \nL = O\bigl( (\numS / \fr)^2 \log^2 (\numS / \fr) \bigr).
    \end{math}

    We now run the above algorithm with exponential search on the
    parameter $t$ used in the \LP from $1$ up to $O( \numS /\fr)$, see
    \Eqref{lp:t}, stopping as soon as the algorithm succeeds. The
    number of separation oracle calls performed in each attempt to
    solve the \LP is $\nL^{O(1)} $. Each such attempt
    involves computing the weak cuttings, which can be done in
    $O( t^3)$ time. Verifying that no face contains too many points of
    any set $\P_j$, can be done by preprocessing the arrangement for
    point-location queries. This takes $O( \nP \log t)$ time for
    point locations per rounding attempt.
    Overall, this results in running time
    $( \nP + \numS / \fr)^{O(1)} $.

    To get an improved running time, we have to avoid performing the
    point-location stage at each rounding attempt. To this end,
    observe that the weak cuttings only adds vertical lines that
    passes through vertices of $\ArrX{\L}$. Let $\L'$ be this set of
    vertical lines, and compute the arrangement $\ArrX{\L \cup
       \L'}$. This arrangement has $( \numS / \fr)^{O(1)}$ faces, an
    it is enough to compute for each face in this arrangement how many
    points of $\P_j$ falls into this face, as all polygons/faces
    considered by the above algorithm are disjoint union of such basic
    faces. Reducing $\P_1, \ldots, \P_\numS$ in this way takes
    $O\bigl( ( \numS / \fr)^{O(1)} + \nP \log ( \numS / \fr) )$
    time. From this point on, we have that the ``reduced'' point sets
    have total size $ ( \numS / \fr)^{O(1)}$,
    and thus each rounding
    attempt can be performed in
    $ (\numS/\fr)^{O(1)}$ time.~
\end{proof}

\section{Approximation algorithm for the guarding problem}
\seclab{g:convex}

The input instance $\Inst$ is made of $\numS$ point sets
$\P_1, \ldots, \P_\numS$ in $\Re^2$ not necessarily distinct, and
fractions $\fr_1, \ldots, \fr_\numS \in (0, 1]$. Let
$\fr = \min_i \fr_i$ and $\nP_i = \cardin{\P_i}$, for
$i\in \IRX{\numS}$, and let $\nP = \sum_i \nP_i$. As a reminder, the
goal is to compute the smallest set of points $\Net$ that is
simultaneously a weak $\fr_i$-net for every point set $\P_i$, that is
for any convex set $\region$ in $\Re^2$ with
$\cardin{\P_i \cap \region} \geq \fr_i \nP_i$, we have that
$\Net \cap \region \neq \emptyset$.  Let $\P = \cup_{i=1}^\numS \P_i$.

\subsection{Reducing the number of candidate points}

Given an instance $\Inst$ of the guarding problem, we again consider
the set of all convex sets $\Polygons$, and focus on the set of bad
polygons $\BPolygons$.  As a reminder, $\pg \in \Polygons$ is
\emphw{bad}, if there is an index $j \in \IRX{\numS}$, such that
$| \P_j\cap \pg | \ge \fr_j |\P_j|$.  We make a few observations about
bad polygons:
\begin{compactenumI}
    \smallskip%
    \item Since we care only with how a convex set interact with the
    set $\P$, it suffices to restrict $\BPolygons$ to the set of bad
    convex polygons that are the convex hull of some subset of points
    of $\P$. We denote this set by
    $\BRP= \BPolygons( \P) = \Set{ \CHX{\pg \cap \P}}{ \pg \in
       \BPolygons}$.

    \smallskip%
    \item Any solution $\Net$ stabs all the polygons in $\BRP$.

    \smallskip%
    \item Thus, consider the arrangement $\ArrX{\BRP}$. Any point
    $\p   \in \Net$, can be moved to a vertex of the face of this
    arrangement that contains it, and it would still stab the same
    polygons of $\BRP$. Since the vertices are defined
    by edges of convex hulls of points in $\P$, we can restrict $\Net$
    to be a subset of the vertices of this arrangement.

    \smallskip%
    \item \itemlab{guard:locations} Specifically, it suffices to
    restrict to point guards that lie on the intersection of two line
    segments joining pairs of points of $\P$. There are $O(|\P|^2)$
    such line segments and thus $O(|\P|^4)$ points defined this way.
    Let $\Q$ be the set of all such points.
\end{compactenumI}

\subsection{The implicit \LP}

The associated \LP for hitting set of the polygons $\BRP$ with at most
$t$ points is:
\begin{align*}
  \valX{\Q} =& \sum_{\p \in \Q} x_\p \leq t
  \\
             &1 \geq x_\p \geq 0
             & \forall \p \in \Q
  \\
  \valX{\pg} = &\sum_{\p \in \Q \cap \pg} x_\p \geq 1
             & \forall \pg \in \BRP. \tag{**}
\end{align*}

\subsection{The rounding scheme}

We will use following result of Rubin \cite{r-ibwen-18} about weak
$\eps$-nets.

\begin{theorem}
    Let $\P$ be a set of $\nP$ points in $\Re^2$, and let
    $\eps \in (0,1)$ be a parameter.  For any $\alpha> 0$, one can
    compute a weak $\eps$-net $S$ of size $O(\eps^{-3/2-\alpha})$ in
    time $\widetilde{O}(\nP^2/\sqrt{\eps})$, where $\widetilde{O}$
    hides polylogarithmic factors in $\nP$.
\end{theorem}

Let $\alpha = \valX{\Q}$.  The above construction applies for discrete
point sets, but we can apply the construction to the points of $\Q$ by
including point $\p\in \Q$ with multiplicity
$v'(\p) = \floor{ 4|\Q| \valX{\p}}$ and compute a weak $\eps$-net
$\W \subseteq \Q$, with $\eps = 1/(4\alpha)$, of size
$O(\alpha^{3/2+\alpha})$ for some constant $\alpha> 0$.  This set can
be constructed in $\widetilde{O}(|Q|\alpha) = \widetilde{O}(\nP^8\alpha^{3/2})$ time.  Observe
that for any multiset of the points $S$,
$4|\Q|\valX{S} \ge v'(S) \ge 4|\Q|\valX{S} - |S|$.  Consider a convex
set $\pg$ with $v(\pg) = v(\pg \cap \Q) \geq 1/2 \geq 2 \eps \alpha$,
and observe that
\begin{equation*}
    v'(\pg\cap \Q)
    \geq
    4|\Q|\valX{\pg\cap \Q} - |\pg\cap \Q|
    \geq
    8 \eps \alpha |\Q| - |\pg \cap \Q|
    \geq
    |\Q|
    =
    \eps \cdot 4 |\Q| v(\Q)
    \geq
    \eps v'(\Q).
\end{equation*}
This implies that any convex set $\pg \in \BRP$ that contains no
points of $\W$ must have $v(\pg) < 1/2$.  The idea is now to test if
there exists a convex polygon $\pg \in \BRP$ that contains no point of
$\W$.  If no such polygon exists, then we have found the desired set
of guards (i.e., we successfully rounded the given LP solution).
Otherwise, we found a polygon $\pg \cap \W =\emptyset$ (i.e.,
$\valX{\pg} < 1/2$) -- namely, we found a constraint belonging to (**)
that is being violated.

\subsubsection{Searching for a bad polygon}

This step was easier for \probref{prob} as we only needed to check
every face of the arrangement of the computed lines.  However, here
there are exponentially many (canonical) convex sets that avoid the
set $\W$ of guards that need to be checked.
\begin{lemma}
    \lemlab{bad:polygon}%
    Given a set $\P$ of $\nP$ points, and another set $\W$ of at most
    $\nP$ points, one can decide, in $O(\nP^4 \log \nP)$ time, if
    there exists a closed convex set $\pg$ that satisfies
    $|\P \cap \pg| \geq \fr \nP$ and $|\W \cap \pg| = 0$.
\end{lemma}
\begin{proof}
    For the simplicity of exposition, we assume the $x$-values of all
    the points under consideration are all distinct -- this can be
    ensured by slightly perturbing the points.  Arguing as above, it
    suffices to consider only polygons $\pg$ that are formed by the
    convex hull of some subset of points in $\P$. Consider two
    segments with endpoints in $\P$ and their vertical decomposition
    -- there might be at most one vertical trapezoid has these two
    segments as a floor and ceiling segments. Let $\VTA$ be the set of
    all such trapezoids. Clearly, the set $\VTA$ can be computed in
    $O( \nP^4)$ time. Furthermore, using simplex range searching, one
    can count for each such trapezoid $\trap$ how many points of $\P$
    it contains (ignoring say points that lie on its right wall),
    denoted by $w(\trap)$, and how many points of $\W$ it
    contains. With $O(\nP^{2+o(1)})$ preprocessing, such queries can
    be answered in $O( \log \nP)$ time \cite{csw-qoubs-92}. Let $\VT$
    be the set $\VTA$ after we remove from it all the trapezoids that
    contains any point of $\W$. The set $\VT$ can be computed in $O(
    \nP^4 \log \nP)$ time, and has size $O(\nP^4)$.

    Two trapezoids $\trap_1, \trap_2 \in \VT$ that share a vertical
    wall are \emphw{compatible}, if there is a polygon
    $\pg \in \Polygons$ such that $\pg \cap \slabX{\trap_1} = \trap_1$
    and $\pg \cap \slabX{\trap_2} = \trap_2$, where $\slabX{\trap_i}$
    is the minimal vertical strip containing $\trap_i$. Note that this
    is a local condition and can be checked in constant time.

    We create a \DAG $\G$ over $\VT$, where an edge
    $\trap_1 \rightarrow \trap_2$ is in $\G$ if $\trap_1$ and
    $\trap_2$ are compatible, and $\trap_1$ is to the left of
    $\trap_2$. By the assumption of unique $x$-coordinates, $\trap_1$
    and $\trap_2$ must share either the bottom or top supporting
    lines.
    The \DAG $\G$ has $O(\nP^4)$ vertices, and
    potentially there are $O(\nP)$ out going edges from each vertex --
    as by assumption two adjacent compatible trapezoids changes only either the
    floor or ceiling supporting segment.  This results in a graph $\G$
    with $O(\nP^5)$ edges. However, by adding special entrance
    vertices to each trapezoid, and chaining them by slope of the
    changing segment, one can reduce the number of edges to
    $O(\nP^4)$. The graph $\G$ can be computed in $O(\nP^4 \log \nP)$ time.

    Note, that any convex polygon in $\BRP$ corresponds to a maximal
    path in the \DAG $\G$. A trapezoid $\trap \in \VT$ is a
    \emphi{start} (resp. \emphi{final}) trapezoid if its left (resp.,
    right) wall is a vertex. Note, that the leftmost (resp.,
    rightmost) trapezoid in any polygon of $\Polygons$ must be a start
    (resp., final) trapezoid.

    The problem thus reduces to computing the longest
    path in the DAG $\G$ with vertex weights $w(\trap)$ for each
    $\trap \in \VT$. This can be done with a standard dynamic program
    in time linear in the size of the DAG, which takes
    $O(\nP^4)$ time.
\end{proof}

\subsection{The result}

The above algorithm provides us with a procedure for computing a bad
polygon if it exists for the currently suggested solution -- namely,
we can use it in the round-and-cut framework.

\begin{theorem}
    \thmlab{g:convex:polygons}%
    Given an instance
    $\Inst =(\P_1,\fr_1, \ldots, \P_\numS, \fr_\numS)$ of the
    \probrefY{prob:guarding}{guarding problem} in the plane of size
    $\nP = \sum_i |\P_i|$, with $\numS$ sets and $\fr = \min_i \fr_i$,
    one can compute a set $\W$ of $O(\opt^{3/2+\alpha})$ points, for
    any fixed $\alpha> 0$, such that $\W$ is a weak $\fr_i$-net for
    $\P_i$ for all $i\in \IRX{\numS}$. That is, for any convex polygon
    $\body$, and any $j \in \IRX{\numS}$, if
    $|\body \cap \P_i| > \fr_i |\P_i|$ then
    $\body \cap W \neq \emptyset$.  The algorithm has running time
    polynomial in $\nP$, assuming $\fr_1, \ldots, \fr_\numS \geq 1/\nP$.
\end{theorem}
\begin{proof}
    Observe that $m$ is a naive upper bound on the number of guards
    needed, as we can simply guard all the points.
    We can get better upper bounds by taking weak $\eps$-nets
    of each class of points, but this is not needed.

    We restrict our attention to our candidate set $\Q$ with
    $\nL = O(\nP^4)$ points by \itemref{guard:locations}.  We now run the
    round-and-cut algorithm with exponential search on $t$ from $1$ up
    to $m$, stopping as soon as the algorithm
    succeeds.  The number of separation oracle calls performed in each
    attempt to solve the \LP is $\nL^{O(1)} = \nP^{O(1)} $.  Each rounding
    step takes $\widetilde{O}(\nP^8 \opt^{3/2})$ where $\opt$ is at most $\nP$.
    The total running time is $\nP^{O(1)}$.
\end{proof}

\section{Conclusions}

We revisited the natural geometric divide-and-conquer
\probrefY{prob}{reduction problem}. In the process we introduced a new
kind of weak $\eps$-nets for corridors (i.e., weak $\eps$-cutting). We
presented a non-trivial construction that provides such nets of size
$\tldO(1/\eps^{3/2})$. Using this construction of nets as a rounding
scheme, used within the round-and-cut framework, we were able to get a
$\tldO( \sqrt{\opt})$-approximation to the optimal solution, where
$\opt$ was the size of the optimal solution. While this approximation
quality is somewhat ``underwhelming'' it is still a significant
improvement when $\opt$ is small (say a constant), where previously
only logarithmic approximation was known, and previous approaches
seems unlikely to lead to a sublogarithmic approximation in the
general case.

We then solved the dual problem of guarding a point set against convex
regions by inserting guards, except that in this case in addition to
the rounding provided by the ``standard'' weak $\eps$-net, we had to
use dynamic programming to find bad convex regions if they exists.

There are numerous open problems for further research raised by our
work. The first one is improving the approximation quality even
further. Secondly, further improving the size of the construction of
weak $\eps$-nets for corridors (i.e., weak $\eps$-cutting for
lines). The bound we get is mysteriously very similar to the best
known bound for the weak $\eps$-net for points (for convex
regions). It is natural to further investigate this
connection. Ultimately, improving and simplifying Rubin's construction
in 2d seems like a worthy problem for further research.

Beyond that, this work emphasize the ``rounding is approximation''
approach. It is natural to wonder if there are other natural geometric
problems where better rounding is possible because of the geometry,
which would lead to better approximation algorithms.

\BibTexMode{%
   \bibliographystyle{alpha}
   \bibliography{weak_cuttings}
}%
\BibLatexMode{\printbibliography}

\end{document}